\date{}

\documentclass[11pt, letterpaper]{article}
\usepackage{amssymb}
\usepackage{fullpage}
\usepackage[noadjust]{cite}
\usepackage[noend]{algpseudocode}
\usepackage{times}

\usepackage{amsfonts,amsmath,color,float,graphicx,verbatim,url, amsthm}

\usepackage{algorithm}

\usepackage{enumerate}

\newtheorem{theorem}{Theorem}
\newtheorem{lemma}[theorem]{Lemma}

\usepackage{graphicx}  

 \usepackage{amsfonts,amsmath} 

 \usepackage{hyperref}
 
 \usepackage{varwidth}

 \usepackage{url}
\newcommand{\stream}{{\sigma}}
\newcommand{\length}{m}
\newcommand{\uni}{\mathcal{U}}
\newcommand{\cV}{\mathcal{V}}
\newcommand{\cP}{\mathcal{P}}
\newcommand{\eps}{\epsilon}
\newcommand{\eat}[1]{}

\begin{document}

\title{Stream Verification}

\author{Justin Thaler\thanks{Yahoo Labs}}

\maketitle 

\begin{abstract}
We survey models and algorithms for stream verification. 
\end{abstract}

\section{Problem Definition}

Stream verification is concerned with the following setting. 
A computationally limited client wants to compute some property of a massive input, but lacks the resources to store even a small fraction of the input, and hence cannot perform the desired computation locally. The client therefore accesses a powerful but untrusted service provider (e.g., a commercial cloud computing service), who not only performs the requested computation, but also proves that the answer is correct.
An array of closely related models have been introduced to capture this scenario. The following section provides a unified presentation of these models, emphasizing their common features before delineating their differences. 

\medskip
\noindent \textbf{Stream Verification Model.}
Let $\stream = \langle a_1, a_2, \dots , a_\length \rangle$ be a data stream, where each $a_i$ comes from a data universe $\uni$ of size $n$, and let $F$ be a function mapping data streams to a finite range $\mathcal{R}$. 
A stream verification protocol for $F$ involves two parties: a \emph{prover} $\cP$, and a (randomized) \emph{verifier} $\cV$. The protocol consists of two stages:  a stream observation stage and a proof verification stage. 
 
In the stream observation stage, $\cV$ processes the stream $\stream$, subject to the 
standard constraints of the data-stream model, i.e., sequential access to $\stream$ and limited memory. 
In the proof verification stage, $\cV$ and $\cP$ exchange a sequence of one or more messages, and afterward $\cV$ outputs a value $b$. $\cV$ is allowed to output a special symbol $\perp$ indicating a rejection of $\cP$'s claims.  
Formally, $\cV$ constitutes a stream verification protocol if the following two properties are satisfied.

\begin{itemize}
\item \textbf{Completeness}: There is some prover strategy $\cP$ such that, for all streams $\stream$, the probability that $\cV$ outputs $F(\stream)$ after interacting with $\cP$ is at least $2/3$. 
\item \textbf{Soundness}: For all streams $\stream$ and all prover strategies $\cP$, the probability that $\cV$ outputs a value not in $\{F(x), \perp\}$ after interacting with $\cP$ is at most $\eps \leq 1/3$.
\end{itemize}

Here, the probabilities are taken over $\cV$'s internal randomness. The constants $2/3$ and $1/3$ are not essential and are chosen by convention. The parameter
$\eps$ is referred to as the \emph{soundness error} of the protocol.

\medskip
\noindent \textbf{Costs.}
There are five primary costs in any stream verification protocol: (1) $\cV$'s space usage, (2) the total communication cost, (3) $\cV$'s runtime, (4) $\cP$'s runtime, and (5) the number of messages exchanged. 

\medskip
\noindent \textbf{Differences Between Models.}
There are three primary differences between the various models of stream verification that have been put forth in the literature. The first is whether the soundness condition is required to hold against all cheating provers (such protocols are called \emph{information-theoretically} or \emph{statistically} sound), or only against cheating provers that run in polynomial time (such protocols are called \emph{computationally} sound). The second is the amount and format of the interaction allowed between $\cP$ and $\cV$.  The third is the temporal relationship between the stream observation and proof verification stage --- in particular, several models permit $\cP$ and $\cV$ to exchange messages before and during the stream observation stage, and sometimes permit the prover's messages to depend on parts of the data stream that $\cV$ has not yet seen. 
In general, more permissive models allow a larger class of problems to be solved efficiently, but may yield protocols that are less realistic.

\medskip
\noindent \textbf{Summary of Models.}
The \emph{annotated data streaming} (ADS) model, introduced by Chakrabarti et al. \cite{icalp}, is non-interactive: $\cP$ is permitted to send a single message to $\cV$, with no communication allowed in the reverse direction. Technically, this model permits the contents of $\cP$'s message to be interleaved with the stream, in which case each bit of $\cP$'s message may be viewed as an ``annotation'' associated with a particular stream update. However, for most ADS protocols that have developed, $\cP$'s message does not need to be interleaved with the stream.
Chakrabarti et al. \cite{icalp} distinguish between two kinds of ADS protocols: \emph{prescient} protocols, in which the annotation sent at any given time can depend on parts of the data stream that $\cV$ has not yet seen, and \emph{online} protocols, which disallow this kind of dependence.

\emph{Streaming interactive proofs} (SIPs) extend the ADS model to allow the prover and verifier to exchange many messages. This model was introduced by Cormode et al. \cite{vldb}.

The \emph{Arthur-Merlin streaming model} was introduced by Gur and Raz \cite{gur}: this model is equivalent to a restricted class of SIPs, in which $\cV$ is only allowed to send a single message to $\cP$ (which must consist entirely of random coin tosses, in analogy with the classical complexity class $\textsc{AM}$), before receiving $\cP$'s reply. 

The model of \emph{streaming delegation} was introduced by Chung et al. \cite{chung}, and corresponds to SIPs that only satisfy computational, rather than information-theoretic, soundness.

\eat{
\begin{table}
\centering
\begin{tabular}{|c|c|c|c|}$\cP$
\hline
Model Name & Soundness Type & Number of & Message \\
 & Type & Messages Allowed & Format Restrictions\\
\hline
Annotated Data Streams \cite{icalp} & Information-Theoretic & 1 & None\\
\hline
Streaming Interactive Proofs \cite{vldb} & Information-Theoretic & Arbitrary & None\\
\hline
Arthur-Merlin Streaming \cite{gur} & Information-Theoretic & 2 & $\cV$'s message to $\cP$\\
& & &  must consist only of coin tosses\\
\hline
Streaming Delegation \cite{chung} & Computational & 2 & None\\
\hline
\end{tabular}
\caption{Model Summary.}
\end{table}
}

\section{Key Results}
Obtaining exact answers even for basic problems in the
standard data streaming model is impossible using $o(n)$ space. In contrast, stream verification
protocols with $o(n)$ space and communication costs have been developed for
(exactly solving) a wide variety of problems. Many of these protocols have adapted powerful algebraic techniques 
originally developed in the classical literature on interactive proofs, particularly the \emph{sum-check} protocol of Lund et al. \cite{lfkn}.
All of the protocols described in Sections \ref{sec:ads}-\ref{sec:cs} apply even to streams in the strict turnstile update model, where universe items can be deleted as well as inserted. The protocols in Section \ref{sec:perupdate} do not support deletions. Unless stated otherwise, all protocols described in this survey are online, i.e., the honest prover's message at any given time does not depend on parts of the data stream that $\cV$ has not yet seen.

\subsection{Annotated Data Streams}
\label{sec:ads}
Chakrabarti et al. \cite{icalp} showed that prescient ADS protocols can be exponentially more powerful than online ones for some problems. 
For example, there is a prescient ADS protocol with logarithmic space and communication costs for computing the median of a sequence of numbers: $\cP$ sends $\cV$ the claimed median $\tau$ at the start of the stream, and while observing the stream, $\cV$ checks that $|\{j:a_j < \tau\}| \leq \length/2$, and $|\{j:a_j > \tau\}| \leq \length/2$, which can be done using an $O(\log m)$-bit counter. Meanwhile, Chakrabarti et al. \cite{icalp} proved that any online protocol for $\textsc{Median}$ with communication
cost $h$ and space cost $v$ requires $h \cdot v = \Omega(n)$, and gave an online ADS protocol achieving these communication--space tradeoffs up to logarithmic factors.

Chakrabarti et al. \cite{icalp} also gave online ADS protocols achieving identical tradeoffs between space and communication costs for problems including \textsc{Frequency Moments} and \textsc{Frequent Items}, and used a lower bound  due to Klauck \cite{Klauck03} on the \emph{Merlin-Arthur communication complexity} of the \textsc{Set-Disjointness} function to show that these tradeoffs are optimal for these problems, even among prescient protocols. 
Subsequent work \cite{esa, semistreaming} gave similarly optimal online ADS protocols for several more problems, including maximum matching and counting triangles in graphs, and matrix-vector multiplication. Chakrabarti et al. \cite{soda} gave optimized protocols for streams whose
length $\length$ is much smaller than the universe size $n$.

 
\subsection{Streaming Interactive Proofs}
Cormode et al. \cite{vldb} showed that several general protocols from the classical literature on interactive proofs can be simulated in the SIP model. 
In particular, this includes a
powerful, general-purpose protocol due to Goldwasser, Kalai, and Rothblum \cite{gkr} (henceforth, the GKR protocol). Given any problem 
computed by an arithmetic or Boolean circuit of polynomial size and polylogarithmic depth,
the GKR protocol requires only polylogarithmic space and communication while using polylogarithmic
rounds of verifier--prover interaction. This yields SIPs for exactly solving many basic streaming problems with
polylogarithmic space and communication costs, including \textsc{Frequency Moments}, \textsc{Frequent Items}, and \textsc{Graph Connectivity}.
Cormode et al. \cite{vldb} also gave optimized protocols for specific problems, including \textsc{Frequency Moments} (see the Detailed Example below).

Gur and Raz \cite{gur} gave an Arthur-Merlin streaming protocol for the \textsc{Distinct Elements} problem 
with communication cost $\tilde{O}(h)$ space cost $\tilde{O}(v)$ for any $h, v$ satisfying $h \cdot v \geq n$, where the $\tilde{O}$ notation hides factors that are polylogarithmic in $n$. Klauck and Prakash \cite{prakashnew}
extended this protocol to give an SIP for \textsc{Distinct Elements} with polylogarithmic space and communication costs and logarithmically
many rounds of prover--verifier interaction. 

Chakrabarti et al. \cite{suresh} gave \emph{constant-round} online SIPs with logarithmic space and communication costs for many problems, including \textsc{Index}, \textsc{Range-Counting}, and \textsc{Nearest-Neighbor Search}. These protocols are exponentially more efficient than what can be achieved by constant-round online SIPs in which $\cV$'s messages to the prover are independent of the input (such as is required in the Arthur-Merlin streaming model of \cite{gur}). 
For classical interactive proofs where the verifier is not restricted to be streaming, allowing $\cV$'s messages to depend on the input does not yield analogous efficiency improvements: Goldwasser and Sipser \cite{goldwassersipser} showed that any interactive proof can be simulated with a polynomial blowup in all costs by an interactive proof in which $\cV$'s messages to $\cP$ consist entirely of random coin tosses.  



\subsection{Computationally Sound Protocols}
\label{sec:cs}
Computationally sound protocols may achieve properties that are unattainable in the information-theoretic setting. For example, they
typically achieve reusability, allowing the verifier to use the same randomness to answer many queries. In
contrast, most SIPs only support ``one-shot'' queries, because they require the verifier to reveal secret
randomness to the prover over the course of the protocol.

Chung et al. \cite{chung} combined the GKR protocol with fully homomorphic encryption (FHE) to give reusable
two-message streaming delegation protocols with polylogarithmic space and communication costs for any problem in the complexity class \textsc{NC}. They
 also gave reusable four-message protocols with polylogarithimic space and communication costs for any problem in the complexity class \textsc{P}. 
 
 Papamanthou et al. \cite{eurocryptpaper} gave
improved streaming delegation protocols for a class of low-complexity queries including point queries and range search: these
protocols avoid the use of FHE, and allow the prover to answer such queries in polylogarithmic time. 
(In contrast, protocols based on the GKR protocol \cite{chung, vldb} require the prover to spend time 
quasilinear in the size of the data stream after receiving
a query, even if the answer itself can be computed in sublinear time.) The protocols of \cite{eurocryptpaper} are based on \emph{hash trees}.

\subsection{Models Allowing Linear Total Communication}
\label{sec:perupdate}
The literature on stream verification contains two models that depart from those described above in that they allow $\cP$ and $\cV$ to communicate on each stream update. In particular, they allow a linear amount of total communication, but aim to bound the maximum amount of communication exchanged (and processing time spent) during any stream update.

\medskip
\noindent \textbf{A Statistically Sound Model.} 
Klauck and Prakash \cite{prakash} studied a variant of online SIPs, in which $\cP$ and $\cV$ are allowed to exchange a constant number of messages during each stream update, as long as at most a constant number of machine words are exchanged during each update. They gave protocols in this model for \textsc{Median}, for determining if a matrix is full-rank, and for approximating the longest-increasing subsequence of the stream.

\medskip
\noindent \textbf{A Computationally Sound Model.} Schr{\"o}der and Schr{\"o}der \cite{ccs} introduced a model requiring only computational soundness, which they called the \emph{verifiable data streaming} (VDS) model. This model is targeted at settings in which there are many parties who may wish to verify answers returned by the prover. In the VDS model, $\cV$ is allowed to send a short message to $\cP$ on each stream update, with no communication allowed in the reverse direction. VDS protocols are required to be reusable and publicly verifiable, in the sense that any party who knows $\cV$'s public key can verify any response by the prover. The VDS model does not permit $\cV$ to update the public key on every stream update, as this would require coordination between all parties wishing to verify information returned by the prover. 

There is a trivial VDS protocol for the \textsc{Index} problem, requiring polylogarithmic communication on each stream update. Recall that in the \textsc{Index} problem, the data stream specifies all entries of a vector $\mathbf{x} = (x_1, \dots, x_n) \in \{0, 1\}^n$, followed by an index $i \in [n]$, and the goal to output $x_i$. In the trivial VDS protocol, for every stream update $x_j$, $\cV$ sends to $\cP$ a digital signature of the tuple $(j, x_j)$. After the stream observation phase, $\cP$ sends to $\cV$ the claimed value of $x_i$ along with a valid digital signature of the tuple $(i, x_i)$. The protocol is secure assuming the prover cannot forge valid signatures: under this assumption, the only way $\cP$ can efficiently produce a valid signature of the tuple $(i, x_i)$ is if $\cV$ had previously sent the signature to $\cP$.

Schr{\"o}der and Schr{\"o}der \cite{ccs} gave a VDS protocol that supports {\sc Index} queries in a more general setting in which the stream not only contains entries of the vector $\mathbf{x}$, but also contains \textsc{Update} operations. Here, an \textsc{Update} operation changes the value of a designated entry $x_j$ of $\mathbf{x}$ to a new value $x'_j$. However, in the protocol of \cite{ccs}, each $\textsc{Update}$ operation requires bidirectional communication (one message from $\cP$ to $\cV$, followed by one from $\cV$ to $\cP$), as well as an update to $\cV$'s public key. 

Krupp et al. \cite{krupp} gave related VDS protocols that reduces costs by logarithmic factors.

\subsection{Implementations}
Implementations of the GKR protocol were provided by Cormode et al. in \cite{itcs} and Thaler in \cite{crypto}. Cormode et al. \cite{itcs} also provided 
optimized implementations of several ADS protocols from \cite{icalp, esa}. Thaler et al. \cite{hotcloud} provided parallelized implementations using Graphics Processing Units. 

Qian \cite{sadimp} refined and implemented the streaming delegation protocol for point queries and range search of Papamanthou et al. \cite{eurocryptpaper}.

Implementations of VDS protocols are described by Schr{\"o}der and Simkin \cite{fc} and Krupp et al.  \cite{krupp}.

\section{Detailed Example}
As described in \cite{vldb}, the sum-check protocol can be directly applied to give an SIP for the $k$th frequency moment problem with $\log n$ rounds of prover--verifier iteration,
and $O(\log^2(n))$ space and communication costs.
The sum-check protocol is described in Figure \ref{fig}.

\begin{figure}
\fbox{
\small
\begin{varwidth}{15cm}
\noindent \textbf{Input:} $\cV$ is given oracle access to a $v$-variate polynomial $g$ over finite field $\mathbb{F}$ and an $H \in \mathbb{F}$.\\
\noindent \textbf{Goal:} Determine whether $H=\sum_{(x_1, \dots, x_v) \in \{0, 1\}^v} g(x_1, \dots, x_v)$. \\
\begin{itemize}
\item In the first round, $\cP$ computes the univariate polynomial $$g_1(X_1) := \sum_{x_2, \dots,x_v \in \{0, 1\}^{v-1}} g(X_1,x_2,\dots,x_v),$$ 
and sends $g_1$ to $\cV$. $\cV$ checks that $g_1$ is a univariate polynomial of degree at most $\deg_1(g)$, and that $H=g_1(0) + g_1(1)$,
rejecting if not.

\item $\cV$ chooses a random element $r_1\in \mathbb{F}$, and sends $r_1$ to $\cP$.

\item In the $j$th round, for $1 < j < v $, $\cP$ sends to $\cV$ the univariate polynomial
$$g_j(X_j) = \sum_{(x_{j+1}, \dots, x_{v}) \in \{0, 1\}^{v-j}}	g(r_1,\dots,r_{j-1}, X_j, x_{j+1}, \dots, x_v).$$ 
$\cV$ checks that $g_j$ is a univariate polynomial of degree at most $\deg_j(g)$,
and that $g_{j-1}(r_{j-1}) = g_j(0) + g_j(1)$, rejecting if not. 

\item $\cV$ chooses a random element $r_j \in \mathbb{F}$, and sends $r_j$ to $\cP$.

\item In round $v$, $\cP$ sends to $\cV$ the univariate polynomial
$$g_v(X_v) = g(r_1,\dots,r_{v-1},X_v).$$ $\cV$ checks that $g_v$ is a univariate polynomial of degree at most
$\deg_v(g)$, rejecting if not. 

\item $\cV$ chooses a random element $r_v \in \mathbb{F}$ and evaluates $g(r_1, \dots, r_v)$ with
a single oracle query to $g$. $\cV$ checks that
$g_v(r_v) = g(r_1, \dots, r_v)$, rejecting if not.

\item If $\cV$ has not yet rejected, $\cV$ halts and accepts.
\end{itemize}
\end{varwidth}
}
\caption{Description of the sum-check protocol. $\deg_i(g)$ denotes the degree of $g$ in the $i$th variable.} 
\label{fig}
\end{figure}

\medskip
\noindent \textbf{Properties and Costs of the Sum-check Protocol.}
The sum-check protocol satisfies perfect completeness, and has soundness error $\eps \leq v \cdot \deg(g)/|\mathbb{F}|$,
where $\deg(g) := \max_{i=1}^v \deg_i(g)$ denotes maximum degree of $g$ in any variable (see \cite{lfkn} for a proof). 
There is one round of prover--verifier interaction in the sum-check protocol for each
of the $v$ variables of $g$, and the total communication $O(v \cdot \deg(g))$ field elements.

Note that as described in Figure \ref{fig}, the sum-check protocol assumes that the verifier has oracle access to $g$. However, this will not be the case in applications,
as $g$ will ultimately be a polynomial that depends on the input data stream.

\medskip \noindent \textbf{The SIP for Frequency Moments.}
In the $k$th frequency moments problem, the goal is to output $\sum_{i \in \uni} f_i^k$, where $f_i$
is the number of times item $i$ appears in the data stream $\sigma$. 
For a vector $\mathbf{i} = (i_1, \dots, i_{\log n}) \in \{0, 1\}^{\log n}$, let 
$\chi_{\mathbf{i}}(x_1, \dots, x_{\log n})= \prod_{k=1}^{\log n}
\chi_{i_k}(x_k)$, where $\chi_0(x_k) = 1 -x_k$ and $\chi_1(x_k) = x_k$. $\chi_{\mathbf{i}}$ is the unique multilinear polynomial that maps 
$\mathbf{i} \in \{0, 1\}^{\log n}$
to $1$ and all other values in $\{0, 1\}^{\log n}$ to 0, and it is referred to
as the \emph{multilinear extension} of $\mathbf{i}$. 

For each $i \in \uni$, associate $i$ with a vector $\mathbf{i} \in \{0, 1\}^{\log n}$ in the natural way, and
let $\mathbb{F}$ be a finite field with $n^k \leq |\mathbb{F}| \leq 4 \cdot n^k$.
Define the polynomial $\hat{f}\colon \mathbb{F}^{\log n} \rightarrow \mathbb{F}$ via
\begin{equation}\label{eq} \hat{f} = \sum_{\mathbf{i} \in \{0,1\}^{\log n}} f_{i} \cdot \chi_{\mathbf{i}}. \end{equation}
Note that $\hat{f}$ is the unique multilinear polynomial satisfying the property that $\hat{f}(\mathbf{i}) = f_i$ for all
$\mathbf{i} \in \{0, 1\}^{\log n}$.

The $k$th frequency moment of $\sigma$ is equal to $\sum_{\mathbf{i} \in \{0, 1\}^{\log n}} f_i^k = \sum_{\mathbf{i} \in \{0, 1\}^{\log n}} (\hat{f}^k)(\mathbf{i})$.
Hence, in order to compute the $k$th frequency moment of $\sigma$, 
it suffices to apply the sum-check protocol to the polynomial $g=\hat{f}^k$. 
This requires $\log n$ rounds of prover--verifier interaction, and since the total degree of $\hat{f}^k$ is $k \cdot \log n$, 
the total communication cost is $O(k \log n)$ field elements.
which require $O(k^2 \log^2 n)$ total bits to specify. 

At the end of the sum-check protocol, $\cV$ must compute $g(r_1, \dots, r_{\log n}) = (\hat{f}^k)(r_1, \dots, r_{\log n})$ for randomly chosen $(r_1, \dots, r_{\log n}) \in \mathbb{F}^{\log n}$. It suffices for $\cV$ to evaluate $z:=\hat{f}(r_1, \dots, r_{\log n})$, since $(\hat{f}^k)(r_1, \dots, r_{\log n})=z^k$. The following lemma establishes
that $\cV$ can evaluate $z$ with a single pass over $\sigma$, while storing $O(\log n)$ field elements.

\begin{lemma}
$\cV$ can compute $z=\hat{f}(r_1, \dots, r_{\log n})$ with a single streaming pass over $\sigma$, while storing $O(\log n)$ field elements.
\end{lemma}
\begin{proof}
Given any stream update $a_j \in \uni$, let $\mathbf{a_j} \in \{0, 1\}^{\log n}$ denote the binary vector associated with $a_j$. 
It follows from Equation \eqref{eq} that $\hat{f}(r_1, \dots, r_{\log n}) = \sum_{j =1}^{\length} \chi_{\mathbf{a_j}}(r_1, \dots, r_{\log n})$.
Thus, $\cV$ can compute $\hat{f}(r_1, \dots, r_{\log n})$ incrementally from the raw stream by initializing $\hat{f}(r_1, \dots, r_{\log n}) \gets 0$, and processing each update $a_j$ via:
\[\hat{f}(r_1, \dots, r_{\log n}) \gets  \hat{f}(r_1, \dots, r_{\log n}) + \chi_{\mathbf{a_j}}(r_1, \dots, r_{\log n}).\]
$\cV$ only needs to store $(r_1, \dots, r_{\log n})$ and $\hat{f}(r_1, \dots, r_{\log n})$, which is $O(\log n)$ field elements in total.
\end{proof}

\nocite{aw, prakash, ccs,allspice, bestorder}
\section{Open Problems}
\begin{itemize}
\item Two-message online SIP protocols with logarithmic space and communication costs are known for several functions, including the \textsc{Index} function (cf. \cite{suresh}). It is also known that \emph{existing techniques} cannot yield 2- or 3-message online SIPs of polylogarithmic cost for the \textsc{Median} or \textsc{Frequency Moments} problems.
However, the following is open: exhibit an explicit function $F: \{0, 1\}^n \rightarrow \{0, 1\}$ that cannot be computed by any online two-message SIP with communication and space costs
both bounded above by $h$, for some $h = \omega(\log n)$. See \cite{suresh} for details.
Candidate functions satisfying this property include \textsc{Set-Disjointness} and \textsc{Inner-Product-Mod-2}.
\item For several functions $F \colon \{0, 1\}^n \rightarrow \{0, 1\}$, it is known that
any online ADS protocol for $F$ with communication
cost $h$ and space cost $v$ requires $h \cdot v = \Omega(n)$. This lower bound is tight
in many cases, such as for the \textsc{Index} function (cf. \cite{icalp}). However,
the following is open: exhibit an explicit function that cannot be computed by any online ADS protocol with communication and space costs
 both bounded above by $h$, for some $h=\omega(n^{1/2})$. 
 \item 
 Chakrabarti et al. \cite{soda} and Thaler \cite{semistreaming} have also identified explicit problems which are just as hard for online ADS protocols as they are in the standard data streaming model, in the sense that the sum of $\cV$'s space cost and the communication cost in any online ADS protocol must be at least as large as the space complexity of a standard streaming algorithm for the problem, up to constant or logarithmic factors. However, it is open to exhibit an explicit function $F$ that is just as hard in for \emph{prescient} ADS protocols as it is in the standard streaming model. 
 
 Formally, the following is open: exhibit an explicit function $F$ such that, for any constant $\delta > 0$, $F$ cannot be computed by any prescient ADS protocol with communication and space cost both bounded above by $O(\text{stream}(F)^{1-\delta})$. Here, $\text{stream}(F)$ denotes the minimum space complexity of any streaming algorithm that, for any input stream $\sigma$, outputs $F(\sigma)$ with probability at least $2/3$. Candidate functions include \textsc{Graph Connectivity} and \textsc{Graph Bipartiteness}. 
\end{itemize}

For all three problems, functions satisfying the relevant properties are known to exist by counting arguments. But as stated above, identifying explicit functions satisfying the properties remains open.



%
%


 \bibliographystyle{plain}
 \bibliography{mydatabase}
%


\end{document}